\documentclass[a4paper,conference]{IEEEtran}

\IEEEsettopmargin{t}{30mm}
\IEEEquantizetextheight{c}
\IEEEsettextwidth{14mm}{14mm}
\IEEEsetsidemargin{c}{0mm}

\usepackage[utf8]{inputenc}
\usepackage[T1]{fontenc}
\usepackage{url}
\usepackage{ifthen}
\usepackage{cite}
\usepackage[cmex10]{amsmath} 
\usepackage{amsfonts}
\usepackage{amssymb}

\usepackage{bbm}
\usepackage{bm}
\usepackage{cancel}
\usepackage{ClearSans}

\usepackage{epsfig}

\usepackage{glossaries}
\usepackage{mathrsfs}
\usepackage{mathtools}
\usepackage{pgfplots}

\usepackage{scalerel}
\usepackage[binary-units]{siunitx}
\usepackage{subcaption}

\usepackage{tikz}

\usetikzlibrary{arrows,decorations.markings, shapes.arrows}
\usetikzlibrary{calc,intersections, positioning,graphs,fit}
\usetikzlibrary{shapes.geometric}
\usetikzlibrary{graphs}
\usetikzlibrary{overlay-beamer-styles}
\usetikzlibrary{decorations.text}

\usepackage{enumitem}
\usepackage[draft]{hyperref}

\newcommand{\field}[1]{\mathbb{#1}}
\newcommand{\group}[1]{\mathscr{#1}}

\newcommand{\bv}[1]{\mathbf{#1}} 
\newcommand{\rv}[1]{\mathsf{#1}} 
\newcommand{\Prob}[1]{\ensuremath{\mathsf{P} \left( #1 \right)}} 

\newcommand{\mat}[1]{{#1}}

\newcommand{\simplex}[1]{\Delta_{#1}}

\DeclareMathOperator{\argmin}{argmin}

\newcommand{\set}[1]{\mathcal{#1}} 

\newtheorem{theorem}{Theorem}
\newtheorem{corollary}{Corollary}[section]
\newtheorem{definition}{Definition}
\newtheorem{remark}{Remark}

\newtheorem{lemma}{Lemma}

\newenvironment{proof}[1][Proof]{\noindent\textbf{#1.} }{\ \rule{0.5em}{0.5em}}

\newtheorem{proposition}{Proposition}[section]

\newacronym{awgn}{AWGN}{Additive White Gaussian Noise}

\newacronym{bac}{BAC}{Binary Asymmetric Channel}

\newacronym{bdsib}{BDSIB}{Binary Double-Sided Information-Bottleneck}
\newacronym{bec}{BEC}{Binary Erasure Channel}

\newacronym{bms}{BMS}{Binary Memoryless Symmetric}

\newacronym{bs}{BS}{base station}

\newacronym{bsc}{BSC}{Binary Symmetric Channel}
\newacronym{bscs}{BSCs}{Binary Symmetric Channels}

\newacronym{ceb}{CEB}{Conditional Entropy Bound}

\newacronym{cr}{CR}{Common Reconstruction}

\newacronym{cran}{C-RAN}{cloud radio access network}

\newacronym{cp}{CP}{Central Proccesor}

\newacronym{dmc}{DMC}{Discrete Memoryless Channel}

\newacronym{dmmac}{DM-MAC}{Discrete Memoryless Multiple Access Channel}

\newacronym{dms}{DMS}{Discrete Memoryless Source}

\newacronym{dpi}{DPI}{Data Proccesing Inequality}

\newacronym{dsbs}{DSBS}{Doubly Symmetric Binary Source}

\newacronym{dsib}{DSIB}{Double-Sided Information-Bottleneck}

\newacronym{epi}{EPI}{Entropy Power Inequality}

\newacronym{gdsib}{GDSIB}{Gaussian Double-Sided Information Bottleneck}

\newacronym{gp}{GP}{Gelf'and-Pinsker}

\newacronym{ib}{IB}{Information Bottleneck}
\newacronym{iid}{i.i.d.}{independent and identically distributed}

\newacronym{infcomb}{IC}{Information Combining}

\newacronym{lhs}{LHS}{Left Hand Side}

\newacronym{sgdsib}{SGDSIB}{Scalar Gaussian Double-Sided Information-Bottleneck}

\newacronym{mac}{MAC}{multiple access channel}

\newacronym{mgl}{MGL}{Mrs. Gerber's Lemma}

\newacronym{mi}{MI}{mutual information}

\newacronym{mu}{MU}{Mobile User}

\newacronym{pf}{PF}{Privacy Funnel}

\newacronym{pmf}{\textsf{pmf}}{probability mass function}

\newacronym{ssib}{SSIB}{Single-Sided Information Bottleneck}
\newacronym{sgssib}{SGSSIB}{Scalar Gaussian Single-Sided Information Bottleneck function}
\newacronym{sawgnssib}{SAWGNSSIB}{Scalar AWGN Single-Sided Information Bottleneck function}

\newacronym{snr}{SNR}{Signal to Noise Ratio}

\newacronym{tibo}{TIBO}{Ternary-Input Binary-Output }

\newacronym{tito}{TITO}{Ternary-Input Ternary-Output }

\newacronym{qiqo}{QIQO}{Quaternary-Input Quaternary-Output }

\newacronym{rhs}{RHS}{right hand side}

\newacronym{rv}{RV}{Random Variable}

\newacronym{rssib}{RSSIB}{Reversed Single-Sided Information Bottleneck function}
\newacronym{rsawgnssib}{RSAWGNSSIB}{Reversed Scalar AWGN Single-Sided Information Bottleneck function}

\newacronym{wtc}{WTC}{Wiretap Channel}

\begin{document}
\title{A Class of Nonbinary Symmetric \\  Information Bottleneck Problems}

\author{%
  \IEEEauthorblockN{Michael Dikshtein and Shlomo Shamai (Shitz)}
  \IEEEauthorblockA{Technion–Israel Institute of Technology\\
              Department of Electrical and Computer Engineering,
              Haifa 3200003, Israel\\
              Email: \{michaeldic@campus, sshlomo@ee\}.technion.ac.il}
}



\maketitle

\begin{abstract}
	We study two dual settings of information processing. Let $ \rv{Y} \rightarrow \rv{X} \rightarrow \rv{W} $ be a Markov chain with fixed joint \acrlong{pmf} $ \mathsf{P}_{\rv{X}\rv{Y}} $ and a mutual information constraint on the pair $ (\rv{W},\rv{X}) $. For the first problem, known as Information Bottleneck, we aim to maximize the mutual information between the random variables $ \rv{Y} $ and $ \rv{W} $, while for the second problem, termed as Privacy Funnel, our goal is to minimize it. In particular, we analyze the scenario for which $ \rv{X} $ is the input, and $ \rv{Y} $ is the output of modulo-additive noise channel. We provide analytical characterization of the optimal information rates and the achieving distributions.
\end{abstract}

\section{Introduction}

Let $ (\rv{X},\rv{Y}) $ be a pair of random variables specified by a fixed bivariate distribution $ \mathsf{P}_{\rv{X}\rv{Y}} $, of cardinality $ |\mathcal{X}| = n $, and respectively $ |\mathcal{Y}|=m $. Consider all random variables $ \rv{W} $ satisfying the Markov chain $ \rv{Y} \rightarrow \rv{X} \rightarrow \rv{W} $ subject to a constraint on the mutual information of the pair $ (\rv{X},\rv{W}) $. 
We consider here two extremes of the information processing problem, the  \acrfull{ib}  function and the  \acrfull{pf}.

The \acrshort{ib} optimization problem, introduced by Tishby et al. \cite{Tishby1999}, is defined as
\begin{equation} \label{eq:ssib_problem_definition}
	\begin{aligned}
		& R_{\mathsf{P}_{\rv{X}\rv{Y}}}(C)   \triangleq 
		& & \underset{\mathsf{P}_{\rv{W}|\rv{X}}}{\text{maximize}}
		& &  I(\rv{Y};\rv{W}) \\
		& 
		& & \text{subject to}
		& &  I(\rv{X};\rv{\rv{W}}) \leq C.
	\end{aligned}
\end{equation}
This problem is illustrated in \autoref{figure:ib_block_diagram}. In our study we aim to determine the maximum value and characterize the achieving conditional distribution $ \mathsf{P}_{\rv{W}|\rv{X}}  $  (test channels) of \eqref{eq:ssib_problem_definition} for a class of  symmetric channels $ \mathsf{P}_{\rv{Y}|\rv{X}} $, and constraints $ C $. We adopt here the slightly irregular notations from \cite{Witsenhausen1975} since our results profoundly rely on that work.

The motivation to study such a model is as follows. Consider a latent random variable $ \rv{Y} $, which constitutes the Markov chain $ \rv{Y} \rightarrow \rv{X} \rightarrow \rv{W} $ and represents a source of information. The user observes a noisy version of $ \rv{Y} $, i.e., $ \rv{X} $, and then tries to compress the observed noisy data such that its reconstructed version, $ \rv{W} $, will be comparable under the maximum mutual information metric to the original data $ \rv{Y} $.  Thus, \eqref{eq:ssib_problem_definition} is essentially a remote source coding problem 
\cite{Wolf1970}, choosing the distortion measure as the  logarithmic-loss. Here $ \rv{W} $ represents the noisy version ($ \rv{X} $) of the source ($ \rv{Y} $) with a constrained number of bits ($ I(\rv{X};\rv{W}) \leq C $), and the goal is to maximize the relevant information in $ \rv{W} $ regarding $ \rv{Y} $ (measured by the mutual information between $ \rv{Y} $ and $ \rv{W} $). In the standard \acrshort{ib} terminology, $ I(\rv{X};\rv{W}) $ is referred to as the complexity of $ \rv{W} $, and $ I(\rv{Y};\rv{W}) $ is referred to as the relevance of $ \rv{W} $.

\begin{figure}[t]
	\centering
	\tikzstyle{block} = [draw, rectangle, minimum width = 0.5cm, minimum height = 0.5cm]
\tikzstyle{sum} = [draw,circle]
\tikzstyle{input} = [coordinate]
\tikzstyle{dummy} = [coordinate]
\tikzstyle{output} = [coordinate]
\tikzstyle{amp} = [draw,shape border rotate = 180,regular polygon,regular polygon sides=3]

\begin{tikzpicture}[thick,scale=0.8, every node/.style={scale=0.8}]
	
	\node [input]	(out1)	at 	(0,0) 		{};
	\node [block]	(source) at 	(3,0)		{\begin{tabular}{c} Bivariate \\ Source \\ $ P_{\rv{X}\rv{Y}} $ \end{tabular}};
	\node [block]	(enc2)	at	(7,0)		{\begin{tabular}{c} Stochastic \\ Encoder \\ $ P_{\rv{W}|\rv{X}} $ \end{tabular}};
	\node [output]	(out2)	at	(10,0)		{};

	\coordinate (dummy1y) at (5.5,-1.25);
	\coordinate (dummy2y) at (8.5,-1.25);
	\node (Iyv) at (7,-1.25) {$ I(\rv{X};\rv{W}) \leq C $};
	
	\coordinate (dummy1uv) at (0.75,-1.75);
	\coordinate (dummy2uv) at (9.25,-1.75);
	\node (Iuv) at (5,-1.75) {$ \max I(\rv{Y};\rv{W})  $};
	
	\draw [<-] (out1) node[left] {$\rv{Y}$} --    (source);
	\draw [->] (source) --  node[above] {$\rv{X}$} (enc2)  ;
	\draw [->] (enc2) --  (out2) node[right] {$ \rv{W} $};
	
	\draw[->] (Iyv) -- (dummy1y);
	\draw[->] (Iyv) -- (dummy2y);
	\draw[dashed] (dummy1y) -- (dummy1y|-enc2);
	\draw[dashed] (dummy2y) -- (dummy2y|-enc2);
	
	\draw[->] (Iuv) -- (dummy1uv);
	\draw[->] (Iuv) -- (dummy2uv);
	\draw[dashed] (dummy1uv) -- (dummy1uv|-enc2);
	\draw[dashed] (dummy2uv) -- (dummy2uv|-enc2);
\end{tikzpicture}
	
	\caption{Block diagram of the \acrlong{ib} function.}
	\label{figure:ib_block_diagram}
\end{figure}
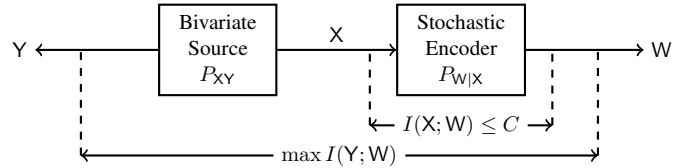

For the particular case where $ (\rv{Y},\rv{X},\rv{W}) $ are discrete random variables, an optimal $ \mathsf{P}_{\rv{W}|\rv{X}} $ can be found by iteratively solving a set of self--consistent equations \cite{Tishby1999}. A generalized Blahuto-Arimoto algorithm \cite{Arimoto1972} 
was proposed to solve those equations.  The optimal test-channel $ \mathsf{P}_{\rv{W}|\rv{X}} $ was characterized using a variation principle in \cite{Tishby1999}. A particular case of deterministic mappings from $ \rv{X} $ to $ \rv{W} $ was considered in \cite{Slonim2002}, and algorithms that find those mappings were described. Unfortunately, since the underlying optimization problem in \eqref{eq:ssib_problem_definition} is not convex, there are no theoretical guarantees for  convergence of the proposed iterative algorithms.


There are two cases for which the solution of  \eqref{eq:ssib_problem_definition} is thoroughly characterized. The first one, considered in \cite{Zaidi2020}, is where the pair $ (\rv{X},\rv{Y})  $ is a \acrlong{dsbs} (\acrshort{dsbs}) with transition probability $ p $. It was shown that the optimal test channel $ \mathsf{P}_{\rv{W}|\rv{X}} $  is a \acrshort{bsc} with transition probability $ h_2^{-1}(1-C) $ where $ h_2(\cdot) $ is binary entropy function and $ h_2^{-1}(\cdot) $ its inverse. The converse can be established by applying \acrlong{mgl}  \cite{Wyner1973}. This setting was also solved as an example in \cite[Section IV.A]{Witsenhausen1975}. The optimality of \acrshort{bsc} test-channel extends  also to a \acrfull{bms} channel \cite[Ch. 4]{Richardson2008}  from $ \rv{X} $ to $ \rv{Y} $, as \cite[Theorem 2]{Sutskover2005} implies.

The second case, first considered in \cite{Chechik2005},  is where $ (\rv{X},\rv{Y}) $ are jointly Gaussian. It was shown that the optimal distribution of $ (\rv{Y},\rv{X},\rv{W}) $ is also jointly Gaussian. The optimality of  the Gaussian test-channel can be proved using conditional \acrlong{epi} \cite[Ch. 2]{Gamal2011}. It can also be established using I-MMSE and Single Crossing Property \cite{Guo2013}. Moreover, under the I-MMSE framework, the proof can be easily extended to Jointly Gaussian Random Vectors $ (\bv{X},\bv{Y}) $  \cite{Bustin2013}.

The \acrshort{ib} method can also be seen as a variation on some  closely related problems in the Information Theory literature. A bound on the conditional entropy for a pair of discrete random variables subject to entropy  constraint has been considered in \cite{Witsenhausen1975} as a method to characterize common information \cite{Gacs1973}. A method  based on convex analysis was proposed to find the achieving distributions and several important examples were given. We will show that the problem addressed in \cite{Witsenhausen1975} is equivalent to \eqref{eq:ssib_problem_definition}. 
The problem of \acrfull{cr} \cite{Steinberg2009} is a different type of source coding with  side-information, a.k.a. Wyner-Ziv coding \cite{Wyner1973}. In \cite{Steinberg2009} the distortion was measured with a log-loss merit, and the encoder is required to perfectly reconstruct decoder's sequence. It can be shown that for the \acrshort{cr}, the resulting single-letter rate-distortion region is equivalent to \acrshort{ib}. 
The problem of  Information Combining \cite{Land2006} was analyzed in the context of  check nodes in LDPC decoding. Two extremes were considered in form of maximization and minimization of mutual information for the binary $ \rv{X} $ setting \cite{Sutskover2005}. It can be shown that the first extreme is equivalent to \acrshort{pf}, while the second recovers the \acrshort{ib} setting.
A recent comprehensive tutorial on the \acrshort{ib} method and related problems is given in \cite{Zaidi2020}. Applications of  \acrshort{ib} methods in Machine Learning are detailed in
\cite{Goldfeld2020}.
Furthermore, the \acrshort{ib} methodology connects to many timely aspects, such as Capital Investment \cite{Erkip1998}, Distributed Learning \cite{Farajiparvar2018}, Deep Learning 
\cite{Amjad2020},
and Convolutional Neural Networks \cite{Yu2020}. 

\begin{figure}[b]
	\centering
	\tikzstyle{block} = [draw, rectangle, minimum width = 0.5cm, minimum height = 0.5cm]
\tikzstyle{sum} = [draw,circle]
\tikzstyle{input} = [coordinate]
\tikzstyle{dummy} = [coordinate]
\tikzstyle{output} = [coordinate]
\tikzstyle{amp} = [draw,shape border rotate = 180,regular polygon,regular polygon sides=3]

\begin{tikzpicture}[thick,scale=0.8, every node/.style={scale=0.8}]
	
	\node [input]	(out1)	at 	(0,0) 		{};
	\node [block]	(source) at 	(3,0)		{\begin{tabular}{c} Bivariate \\ Source \\ $ P_{\rv{X}\rv{Y}} $ \end{tabular}};
	\node [block]	(enc2)	at	(7,0)		{\begin{tabular}{c} Stochastic \\ Encoder \\ $ P_{\rv{W}|\rv{X}} $ \end{tabular}};
	\node [output]	(out2)	at	(10,0)		{};

	\coordinate (dummy1y) at (5.5,-1.25);
	\coordinate (dummy2y) at (8.5,-1.25);
	\node (Iyv) at (7,-1.25) {$ I(\rv{X};\rv{W}) = C $};
	
	\coordinate (dummy1uv) at (0.75,-1.75);
	\coordinate (dummy2uv) at (9.25,-1.75);
	\node (Iuv) at (5,-1.75) {$ \min I(\rv{Y};\rv{W})  $};
	
	\draw [<-] (out1) node[left] {$\rv{Y}$} --    (source);
	\draw [->] (source) --  node[above] {$\rv{X}$} (enc2)  ;
	\draw [->] (enc2) --  (out2) node[right] {$ \rv{W} $};
	
	\draw[->] (Iyv) -- (dummy1y);
	\draw[->] (Iyv) -- (dummy2y);
	\draw[dashed] (dummy1y) -- (dummy1y|-enc2);
	\draw[dashed] (dummy2y) -- (dummy2y|-enc2);
	
	\draw[->] (Iuv) -- (dummy1uv);
	\draw[->] (Iuv) -- (dummy2uv);
	\draw[dashed] (dummy1uv) -- (dummy1uv|-enc2);
	\draw[dashed] (dummy2uv) -- (dummy2uv|-enc2);
\end{tikzpicture}
	\caption{Block diagram of Privacy Funnel.}
	\label{figure:privacy_funnel_block_diagram}
\end{figure}
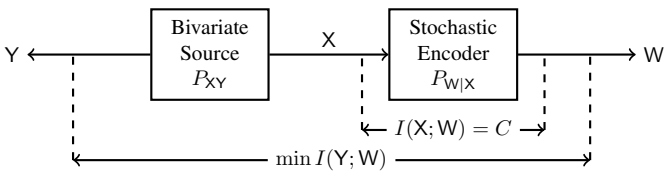

The \acrshort{pf}, which was first introduced in \cite{PinCalmon2017}, is a dual problem to the \acrshort{ib} method. In contrast to \acrshort{ib} problem, the goal in \acrshort{pf}, illustrated in \autoref{figure:privacy_funnel_block_diagram}, is to minimize $ I(\rv{Y};\rv{W}) $ over all test-channels $ P_{\rv{W}|\rv{X}} $ subject to $ I(\rv{X};\rv{W}) = C $. To be more formal, the \acrshort{pf} function, $ \reflectbox{R}:  [0,H(\rv{X})] \to \mathbb{R}_+$ is defined as
\begin{equation} \label{eq:sspf_problem_definition}
	\begin{aligned}
		& \reflectbox{R}_{\scaleto{\mathsf{P}_{\rv{X}\rv{Y}}}{4pt}}(C) \triangleq 
		&&  \underset{\mathsf{P}_{\rv{W}|\rv{X}}}{\text{minimize}}
		&&  I(\rv{W};\rv{Y}) \\
		&
		& & \text{subject to}
		& & I(\rv{X};\rv{\rv{W}}) \mkern-3mu = \mkern-3mu  C.
	\end{aligned}
\end{equation}
Note that taking the constraint here with equality is essential since inequality constraint (i.e. $ I(\rv{X};\rv{W}) \leq C $)  will induce a trivial solution, i.e. taking $ \rv{X} $ and $ \rv{W} $ independent.

\acrshort{pf} is directly connected to Information Combining \cite{Sutskover2005,Land2006}. For example, if the channel from $ \rv{X} $ to $ \rv{Y} $ is a \acrshort{bms}, then by \cite{Sutskover2005}, $ P_{\rv{W}|\rv{X}} $ is a \acrfull{bec}. 
A rather intriguing for the Gaussian setting, where $ (\rv{X},\rv{Y}) $ are jointly Gaussian, the result is zero, since one can use the channel from $ \rv{X} $ to $ \rv{W} $ to describe the less significant bits of $ \rv{X} $ \cite{Shamai2021}.
Furthermore, the additive noise Helper problem studied in \cite{Bross2019}, is  directly linked to the \acrshort{pf}. By reformulating the former as an information combining problem, the solution follows directly as was shown in \cite{Shamai2021}.

In this work we address the input symmetric nonbinary setting for the \acrshort{ib} and \acrshort{pf} functions. We will find conditions on the bivariate source $ (\rv{X},\rv{Y}) $ for which the stochastic encoder from $ \rv{X} $ to $ \rv{W} $ can be completely characterized, thus extending the binary examples from \cite{Witsenhausen1975,Sutskover2005} and \cite{Zaidi2020}.


\section{Notations and Basic Properties}
We denote by $ \Delta_n $ the $ n $ dimensional probability simplex, $ \bv{q} \in \simplex{n} $ the marginal probability vector of $ \rv{X} $, and $ T $ the transition matrix from $ \rv{X} $ to $ \rv{Y} $, i.e.,
\begin{equation}
	T_{ij} \triangleq \Prob{\rv{Y}=i|\rv{X}=j}, \qquad 1\leq i \leq m, 1\leq j \leq n.
\end{equation}

We further rewrite \eqref{eq:ssib_problem_definition} with explicit dependence on $ \bv{q} $ and $ \mat{T} $ as $ R_T(\bv{q},C) = R(C)  = R_{\mathsf{P}_{\rv{X}\rv{Y}}} (C)$.
The entropy of an $ n $-ary probability vector $ \bv{p} \in \Delta_n $ is denoted by $ h_n(\bv{p}) $.

The following tight cardinality bound  was established in \cite{Asoodeh2020}. It was actually already proved for the corresponding dual problem, namely the  \acrshort{ib} Lagrangian, in \cite{Harremoes2007}. But since $ R_T(\bv{q},C) $ is generally not a strictly convex function of $ C $, the result in \cite{Harremoes2007} cannot be directly applied for our problem \eqref{eq:ssib_problem_definition}.
\begin{lemma}[{\cite[Th. 9]{Asoodeh2020}}] \label{lemma:ssib_cardinality}
	The optimization over $ \rv{W} $ in \eqref{eq:ssib_problem_definition} can be restricted to $ | \mathcal{W} | \leq n $.
\end{lemma}


As we have already mentioned, the \acrshort{ib} function defined in \eqref{eq:ssib_problem_definition} is closely related to the \acrfull{ceb} problem studied in \cite{Witsenhausen1975}, which is given by
\begin{equation} \label{eq:ww1975_problem_definition}
	\begin{aligned}
		& F_T(\bv{q},x)  \triangleq  
		& & \underset{\rv{W} \rightarrow \rv{X} \rightarrow \rv{Y}}{\text{minimize}}
		& &  H(\rv{Y}|\rv{W}) \\
		& 
		& & \text{subject to}
		& & H(\rv{X}|\rv{\rv{W}}) \geq x.
	\end{aligned}
\end{equation}
%

\begin{remark}
	Note that originally in \cite{Witsenhausen1975} the conditional entropy constraint was given with equality, and equivalence to the inequality setting was established in \cite[Theorem 2.5]{Witsenhausen1975}
\end{remark}
It turns out that the aforementioned problem is closely connected to the \acrshort{ib} function.
\begin{proposition} \label{proposition:equivalence_of_ssib_and_ww1975}
	The \acrshort{ib} function  defined in \eqref{eq:ssib_problem_definition} 
	is equivalent to the \acrshort{ceb} function defined in \eqref{eq:ww1975_problem_definition}.
\end{proposition}
The proof is postponed to Appendix \ref{section:proof_equivalence_of_ssib_and_ww1975}.

The latter result implies that we can utilize the properties of $ F_T(\bv{q},x) $ developed in \cite{Witsenhausen1975} for our problem in a straightforward manner, an aspect that we will heavily rely on in \autoref{section:symmetric_ib}.


In a very similar manner to \autoref{proposition:equivalence_of_ssib_and_ww1975}, we can redefine the Privacy Funnel problem defined in \eqref{eq:sspf_problem_definition} as follows.
\begin{equation} \label{eq:sspf_problem_definition2}
	\begin{aligned}
		& \reflectbox{F}_{\scaleto{\mat{T}}{4pt}}(\bv{q},x) \triangleq 
		&&  \underset{\mathsf{P}_{\rv{W}|\rv{X}}}{\text{maximize}}
		&&  H(\rv{Y}|\rv{W}) \\
		&
		& & \text{subject to}
		& & H(\rv{\rv{X}}|\rv{W}) \mkern-3mu = \mkern-3mu  x.
	\end{aligned}
\end{equation}

We have the following characterization of $ \reflectbox{F}_{\scaleto{\mat{T}}{4pt}}(\bv{q},x)  $.
\begin{theorem}
	The function $ \reflectbox{F}_{\scaleto{\mat{T}}{4pt}}(\bv{q},\cdot)  $ is concave on the compact convex domain $ \{x: 0 \leq x \leq h_n(\bv{q})\} $ and for each $ (\bv{q},x) $, the maximum is attained with $ \rv{W} $ taking at most $ n+1 $ values.
\end{theorem}
The proof of this theorem is similar to \cite[Theorem 2.3]{Witsenhausen1975} and is omitted here due to space limitations.

\section{The Symmetric Information Bottleneck} \label{section:symmetric_ib}
In this section we will give a characterization of the achieving conditional distributions and the value of the problem defined in \eqref{eq:ssib_problem_definition} for specific class of  input symmetric channels. 
We begin with the definitions of symmetric group of permutation, symmetry group of  stochastic matrix and input symmetric channel \cite{Witsenhausen1975}.

\begin{definition} \label{definition:symmetric_group}
	Let $ \mathscr{S}_{n} $ denote the representation of the symmetric group of permutation of $ n $ objects by the $ n\times n $ permutation matrices. Let $ \mathscr{S} _n \times  \mathscr{S}_m $ be the representation of the direct product group by the pairs $ (\mat{G},\mat{\Pi}) $, $ G \in \mathscr{S}_n $; $ \Pi \in \mathscr{S}_m $ with the composition $ (\mat{G}_1,\mat{\Pi}_1) (\mat{G}_2,\mat{\Pi}_2) = (\mat{G}_1 \mat{G}_2 ,\mat{\Pi}_1 \mat{\Pi}_2) $.
	
	For an $ m \times n $ stochastic matrix $ T $, (an $ n $ input, $ m $ output channel), let $ \mathscr{G} $ be the set $ \{(\mat{G},\mat{\Pi}) \in \mathscr{S}_n \times \mathscr{S}_m | TG = \Pi T \} $, and let $ \mathscr{G}_i $ ($ \mathscr{G}_o $) be the projections of $ \mathscr{G} $ on the first (second) factor.
	If $ \mat{T} \mat{G}_1 = \mat{\Pi}_1 \mat{T} $, $ \mat{T}G_2  = \mat{\Pi}_2 T$, then $ \mat{T} \mat{G}_1 \mat{G}_2 = \mat{\Pi}_1 \mat{\Pi}_2 T $ which shows that $ \mathscr{G} $, $\mathscr{G} _i$ $, \mathscr{G}_o $ are subgroups of the finite groups $ \mathscr{S}_n \times \mathscr{S}_m $, $ \mathscr{S}_n $, $ \mathscr{S}_m $ respectively. $ \mathscr{G} $ is the symmetry group of $ \mat{T} $, $ \mathscr{G}_i  $ ($ \mathscr{G}_0 $) is the input (output) symmetry group.
	
	The channel defined by $ T $ will be called input (output) symmetric if $ \mathscr{G}_i  $ ($ \mathscr{G}_o $) is transitive (a subgroup of $ \mathscr{S}_n $ is transitive if each element of $ \{1,\dots,n\} $ can be mapped to every other element of $ \{1,\dots,n\} $ by some member of the subgroup). $ T $ is said to be symmetric if both $ \mathscr{G}_i $ and $ \mathscr{G}_o $ are transitive.
\end{definition}

We also define the set of  $ (\bv{q},C) $ for which we will have a complete characterization of the achieving distributions.
\begin{definition}
	Let $ \phi(\bv{p},\lambda) \triangleq h_m(T\bv{p}) - \lambda h_n(\bv{p}) $ and $ \bv{p}^* = \argmin_{\bv{p} \in \simplex{n}} \phi(\bv{p},\lambda) $. We define the following set for any $ \lambda \in [0,1] $ and $ \{ G_\alpha \}_{\alpha=1}^n \in \group{S}_n $:
	\begin{equation} \label{eq:definition_of_the_set_q_C}
		\set{Q} \triangleq \left\{ \mkern-2mu
		(\bv{q},C ) : 
		\bv{q} \mkern-4mu = \mkern-4mu \mkern-2mu \sum_{\alpha=1}^n \mkern-2mu w_a \mat{G}_{\alpha} \bv{p}^*,  \bv{w}  \mkern-4mu \in  \mkern-4mu \simplex{n}, C \mkern-4mu = \mkern-4mu 1 \mkern-4mu - \mkern-4mu h_n(\bv{p}^{\star})  \mkern-4mu
		\right\} \mkern-4mu.
	\end{equation}
\end{definition}

Equipped with this definition we are ready to state our main theorem here.
\begin{theorem} \label{theorem:input_symmetric_transition_matrices}
	Assume that $ T $ is input symmetric stochastic matrix with input symmetry group $ \group{G}_i $ of order $ n $. Then for every $ (\bv{q},C) \in \set{Q} $ defined in \eqref{eq:definition_of_the_set_q_C}, the optimal test-channel from $ \rv{W} $ to $ \rv{X} $ is a modulo-additive channel. 
\end{theorem}

Note if $ \bv{q} $ is uniform over $ n $, then it always in $ \set{Q} $, as taking $ \bv{w} $ to be uniform over $ n $, we obtain
\begin{equation}
	\bv{q} = \sum_{\alpha=1}^n w_a \mat{G}_{\alpha} \bv{p}^* = \frac{1}{n} \sum_{\alpha=1}^n \mat{G}_{\alpha} \bv{p}^* = \bv{u}_n,
\end{equation}
where $ \bv{u}_n $ is an $ n $-ary uniform probability vector. This fact induces the following corollary.
\begin{corollary}
	Assume that $ T $ is input symmetric stochastic matrix with input symmetry group $ \group{G}_i $ of order $ n $ and $ \rv{X} $ is uniformly distributed over $ n $. Then for every $ C \in [0,\log n] $, the test-channel from $ \rv{W} $ to $ \rv{X} $ is a modulo-additive noise channel and $ \rv{W} $ is uniform over $ n $.
\end{corollary}

We will prove \autoref{theorem:input_symmetric_transition_matrices} in Appendix \ref{section:proof_of_circulant_case}. In the meantime, let us consider some special cases.

A particular case for which $ \mat{T} $ is input symmetric, is when the channel from $ \rv{X} $ to $ \rv{Y} $ is a modulo-additive noise channel, i.e., there exist a random variable $ \rv{Z} $, with probability vector $ \bv{z} $ such that $ \rv{Y} = \rv{X} \oplus \rv{Z} $, where $ \oplus $ is modulo $ n $ addition. An equivalent representation of the modulo-additive noise channel is using circulant matrix.
A circulant matrix $ \mat{A} \in M_n(\field{F}) $ \cite[p. 33]{Horn2012} has  the form
\begin{equation}\label{eq:circulant_matrix}
	\mat{A} =
	\begin{pmatrix}
		a_1 & a_2 & & \cdots & a_n \\
		a_n & a_1 & a_2 & \cdots & a_{n-1} \\
		\vdots & \vdots & \ddots & \ddots & \vdots \\
		a_2 & a_3 &  \cdots & a_{n} & a_1
	\end{pmatrix},
\end{equation}
i.e, 
the entries in each row are a cyclic permutation of those in the first.
In this case we have the following corollary.
\begin{corollary} \label{corollary:ssib_modulo_additive}
	If $ \mat{T} = \mat{A}  $ as defined in \eqref{eq:circulant_matrix}, than the modulo additive test channel from $ \rv{W} $ to $ \rv{X} $ achieves  $ R_{\mat{A}}(\bv{q},C) $. In particular, there exists an $ n $-ary random variable $ \rv{V} $, with $ H(\rv{V}) = \log n - C $, such that $ \rv{X} = \rv{W} \oplus \rv{V} $ achieves $ R_{\mat{A}}(\bv{q},C) $.
\end{corollary}

Although this result greatly simplifies the optimization space, it does not give a precise analytical solution to the problem. In the following subsection, we provide an example, for which the achieving distribution and the objective function value can be fully characterized.

\subsection{Hamming Channels}
Let $ T = T_{\alpha} = \alpha I_n + (1-\alpha)n^{-1} E_n $, where $ I_n $ is the $ n \times n $ identity matrix, $ E_n $ the all ones matrix, and $ 0 \leq \alpha \leq 1 $. The channel with transition matrix $ T_{\alpha} $ is called a Hamming channel with parameter $ \alpha $. Note that $ T_{\alpha} $ is in particular a circulant matrix, therefore by \autoref{corollary:ssib_modulo_additive} the optimal channel from $ \rv{W} $ to $ \rv{X} $ is a modulo-additive channel.
Thus, \eqref{eq:ww1975_problem_definition} can be reformulated as follows.
\begin{equation} \label{eq:ssib_hamming_problem_definition}
	\begin{aligned}
		& F_T(\bv{q},x)  \triangleq  
		& & \underset{\bv{v} \in \simplex{n}}{\text{minimize}}
		& &  h_n(\mat{T}_{\alpha} \bv{v}) \\
		& 
		& & \text{subject to}
		& & h_n(\bv{v}) \geq x.
	\end{aligned}
\end{equation}
The optimization problem defined in \eqref{eq:ssib_hamming_problem_definition} is identical to the problem considered in \cite{Witsenhausen1974}. Furthermore, it was solved for the Hamming channel and the achieving distribution was found.
\begin{lemma}[{\cite[Lemma 7]{Witsenhausen1974}}]
	For $ n \times n $ Hamming channel $ \mat{T}_\alpha $ the solution to \eqref{eq:ssib_hamming_problem_definition} is attained for
	\begin{equation}
		\bv{v} = 	\beta \bv{e} + (1-\beta) \bv{u}_n.
	\end{equation}
	where $ \bv{e} $ is any standard basis vector of $ \simplex{n} $.
\end{lemma}
Since $ \bv{v} $ is determined by a single parameter $ \beta $ and satisfies $ h_n(\bv{v}) = \log n - C  $, we can find $ \beta $ explicitly as follows:
\begin{align*}
	& C
	= \log n - h_n(\bv{v}) \\
	&= \frac{n-1}{n} (1-\beta) \log (1-\beta) + \frac{\beta n  + 1- \beta}{n} \log (\beta n + 1-\beta) \\
	& \triangleq g_n(\beta).
\end{align*}
Thus, $ \beta $ can be recovered from $ C $ as $ \beta = g_n^{-1}(C) $.
In summary, we have the following theorem.
\begin{theorem} \label{theorem:ssib_hamming}
	Assume that $ \mat{T} $ is a Hamming channel with parameter $ \alpha $, then $ R_T(\bv{u}_n,C) $ is attained with a Hamming channel with parameter $ \beta = g_n^{-1}(C)   $ and is given by
	\begin{equation}
		R_T \mkern-2mu ( \mkern-2mu \bv{u}_n, \mkern-2mu C \mkern-2mu ) \mkern-5mu =  \mkern-5mu \frac{1 \mkern-4mu + \mkern-4mu (\mkern-2mu n \mkern-4mu - \mkern-4mu 1\mkern-2mu)\alpha \beta}{n} \mkern-2mu  \log  (\mkern-2mu 1 \mkern-2mu + \mkern-2mu ( \mkern-2mu n \mkern-2mu - \mkern-2mu 1)\alpha \beta \mkern-2mu) \mkern-2mu  + \mkern-2mu \frac{1 \mkern-4mu - \mkern-4mu \alpha \beta}{n} \mkern-2mu  \log (1 \mkern-2mu - \mkern-2mu \alpha \beta).
	\end{equation}
\end{theorem}
\begin{figure}[b!]
	\centering
	\begin{subfigure}[t]{0.49\textwidth}
		\centering
		\input{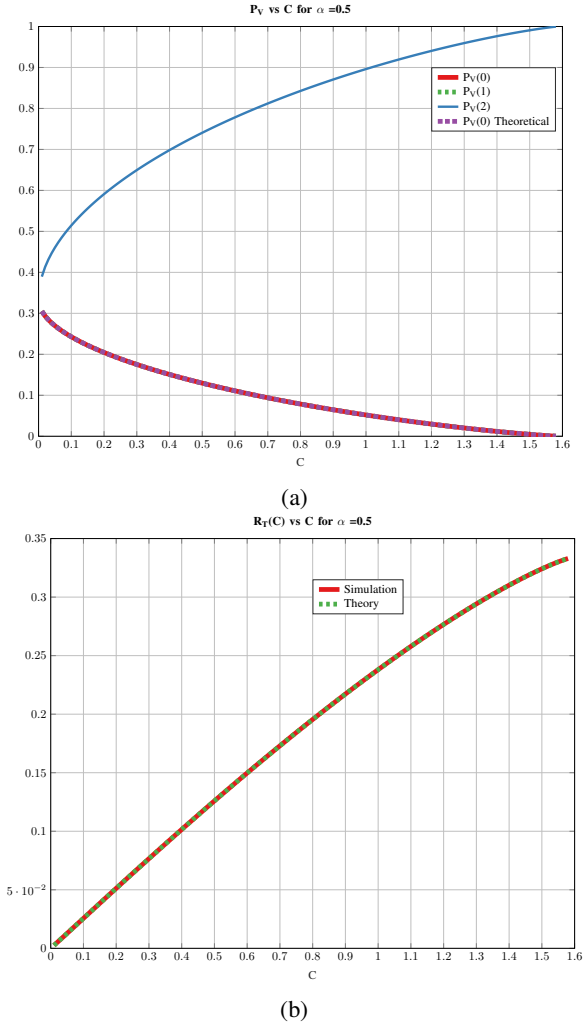}
		\caption{}
	\end{subfigure}
	\begin{subfigure}[t]{0.49\textwidth}
		\centering
%
%
\definecolor{mycolor1}{rgb}{0.89412,0.10196,0.10980}%
\definecolor{mycolor2}{rgb}{0.30196,0.68627,0.29020}%
\begin{tikzpicture}[scale=0.45]
	
	\begin{axis}[%
		width=6.028in,
		height=4.754in,
		at={(1.011in,0.642in)},
		scale only axis,
		xmin=0,
		xmax=1.6,
		xlabel style={font=\color{white!15!black}},
		xlabel={C},
		ymin=0,
		ymax=0.35,
		axis background/.style={fill=white},
		title style={font=\bfseries},
		title={$\text{R}_{\text{T}}\text{(C) vs C for }\alpha\text{ =0.5}$},
		xmajorgrids,
		ymajorgrids,
		legend style={legend cell align=left, align=left, draw=white!15!black,at={(0.5,.9)},anchor=north west}
		]
		\addplot [color=mycolor1, line width=4.0pt]
		table[row sep=crcr]{%
			0.01	0.0025293815781009\\
			0.02	0.00507829446529096\\
			0.03	0.00763648774508985\\
			0.04	0.0102004969933482\\
			0.05	0.0127683388311164\\
			0.06	0.0153386844211512\\
			0.07	0.0179105632394849\\
			0.08	0.0204832259493426\\
			0.09	0.0230560709204644\\
			0.1	0.0256286006798061\\
			0.11	0.0282004202320087\\
			0.12	0.0307714703327662\\
			0.13	0.0333407509308747\\
			0.14	0.0359079557417861\\
			0.15	0.0384736187654851\\
			0.16	0.0410374374957998\\
			0.17	0.043598140539705\\
			0.18	0.0461570329663004\\
			0.19	0.0487127994454841\\
			0.2	0.0512652770444615\\
			0.21	0.0538154537042426\\
			0.22	0.056362046592247\\
			0.23	0.0589053011956104\\
			0.24	0.0614450901371568\\
			0.25	0.0639809297623395\\
			0.26	0.0665137928895887\\
			0.27	0.0690424811690198\\
			0.28	0.071567251583029\\
			0.29	0.0740880024514783\\
			0.3	0.0766046357718184\\
			0.31	0.0791166944422401\\
			0.32	0.0816251737350968\\
			0.33	0.0841288974789731\\
			0.34	0.0866281413910532\\
			0.35	0.089122820965895\\
			0.36	0.0916128537353598\\
			0.37	0.0940977970398813\\
			0.38	0.0965786578231467\\
			0.39	0.099054272576351\\
			0.4	0.101524927102131\\
			0.41	0.103990546434812\\
			0.42	0.106450693093151\\
			0.43	0.108906005060605\\
			0.44	0.111356459559825\\
			0.45	0.113801212982358\\
			0.46	0.116240562952705\\
			0.47	0.118674451616401\\
			0.48	0.121102805748567\\
			0.49	0.123525556597237\\
			0.5	0.125942639823129\\
			0.51	0.128353613907948\\
			0.52	0.130759511945482\\
			0.53	0.133159169754512\\
			0.54	0.135552885706004\\
			0.55	0.137940529482095\\
			0.56	0.140322223816409\\
			0.57	0.142697648824459\\
			0.58	0.145066927249665\\
			0.59	0.147429988116013\\
			0.6	0.149786645878536\\
			0.61	0.152136833334803\\
			0.62	0.154480664777387\\
			0.63	0.156817886953742\\
			0.64	0.159148491258456\\
			0.65	0.161472358325851\\
			0.66	0.163789597283604\\
			0.67	0.166099890574585\\
			0.68	0.168403421687691\\
			0.69	0.170699931407672\\
			0.7	0.172989408783039\\
			0.71	0.175271731938664\\
			0.72	0.177546623143291\\
			0.73	0.179814977125249\\
			0.74	0.182075248572128\\
			0.75	0.184328896378561\\
			0.76	0.186574322659729\\
			0.77	0.188812976872569\\
			0.78	0.191043251618634\\
			0.79	0.193266602700468\\
			0.8	0.195481812947559\\
			0.81	0.197689176368561\\
			0.82	0.19988861408413\\
			0.83	0.202080053918475\\
			0.84	0.204263410856028\\
			0.85	0.206438302661751\\
			0.86	0.20860554609273\\
			0.87	0.210764160417462\\
			0.88	0.212914347884478\\
			0.89	0.215056033216909\\
			0.9	0.217189127812214\\
			0.91	0.219313529682315\\
			0.92	0.221429163131217\\
			0.93	0.223535919071864\\
			0.94	0.225633714561021\\
			0.95	0.227722449012818\\
			0.96	0.229802027956862\\
			0.97	0.231872339124342\\
			0.98	0.233933292069115\\
			0.99	0.235984719609577\\
			1	0.238026691350469\\
			1.01	0.240058926040208\\
			1.02	0.242081363125382\\
			1.03	0.244093588617885\\
			1.04	0.246096394114954\\
			1.05	0.248088456660748\\
			1.06	0.25007079365014\\
			1.07	0.252042282180428\\
			1.08	0.254003742375733\\
			1.09	0.25595427565331\\
			1.1	0.257894021348915\\
			1.11	0.259822839955945\\
			1.12	0.261725368577643\\
			1.13	0.263646846298036\\
			1.14	0.26554200582826\\
			1.15	0.267425907200772\\
			1.16	0.26929784806986\\
			1.17	0.271157935240843\\
			1.18	0.273005983506222\\
			1.19	0.274841825259368\\
			1.2	0.27666489039307\\
			1.21	0.27847612633412\\
			1.22	0.280274196113285\\
			1.23	0.282059259510854\\
			1.24	0.283831115158443\\
			1.25	0.285589520827711\\
			1.26	0.287333871688324\\
			1.27	0.289065056981009\\
			1.28	0.290781297001878\\
			1.29	0.292483840232912\\
			1.3	0.294171260613824\\
			1.31	0.295843637119162\\
			1.32	0.29750065585857\\
			1.33	0.29914158861978\\
			1.34	0.300767231347243\\
			1.35	0.302375680120087\\
			1.36	0.303967751671717\\
			1.37	0.305542729666827\\
			1.38	0.307099761913763\\
			1.39	0.308638281319745\\
			1.4	0.310158120862871\\
			1.41	0.311658745811114\\
			1.42	0.3131393729508\\
			1.43	0.314599681017221\\
			1.44	0.316037755125629\\
			1.45	0.317452707244495\\
			1.46	0.31884670026277\\
			1.47	0.32021513834979\\
			1.48	0.321558277097019\\
			1.49	0.322873689327229\\
			1.5	0.324160758954011\\
			1.51	0.325416981626596\\
			1.52	0.326639977594767\\
			1.53	0.327826576240473\\
			1.54	0.328972904548115\\
			1.55	0.330073337200614\\
			1.56	0.331119959754591\\
			1.57	0.332098083423662\\
			1.58	0.332976242372765\\
		};
		\addlegendentry{Simulation}
		
		\addplot [color=mycolor2, dashed, line width=4.0pt]
		table[row sep=crcr]{%
			0.01	0.00252839801565874\\
			0.02	0.00508232014056453\\
			0.03	0.00763964469958966\\
			0.04	0.0101948699420329\\
			0.05	0.0127608798402528\\
			0.06	0.0153404999201294\\
			0.07	0.0179102309157897\\
			0.08	0.020477445029035\\
			0.09	0.0230498949476667\\
			0.1	0.0256348018251427\\
			0.11	0.028193416312319\\
			0.12	0.0307852234533819\\
			0.13	0.0333447897683112\\
			0.14	0.0359058082259436\\
			0.15	0.0384789184706471\\
			0.16	0.0410444628588036\\
			0.17	0.043587605018621\\
			0.18	0.046155851762427\\
			0.19	0.0487273944362796\\
			0.2	0.0512626151189071\\
			0.21	0.0538076068024016\\
			0.22	0.0563640414215305\\
			0.23	0.0589104309374699\\
			0.24	0.061442508361933\\
			0.25	0.0639715058545089\\
			0.26	0.0665037211070165\\
			0.27	0.069027998091399\\
			0.28	0.0715771768089299\\
			0.29	0.0740948243103896\\
			0.3	0.0765973781295068\\
			0.31	0.0791214266747096\\
			0.32	0.0816309808568236\\
			0.33	0.0841497145381762\\
			0.34	0.0866269360917742\\
			0.35	0.0891150615914666\\
			0.36	0.0916146461703162\\
			0.37	0.0940978448579517\\
			0.38	0.0965718835261031\\
			0.39	0.0990321386176845\\
			0.4	0.10153936045829\\
			0.41	0.103979633401874\\
			0.42	0.106463182724901\\
			0.43	0.108898090631753\\
			0.44	0.111345690123415\\
			0.45	0.11378689425644\\
			0.46	0.116242724193517\\
			0.47	0.118671847618263\\
			0.48	0.121113298355722\\
			0.49	0.123527235900384\\
			0.5	0.125941826814562\\
			0.51	0.128343864967232\\
			0.52	0.130753575164326\\
			0.53	0.133178034401215\\
			0.54	0.135569175036486\\
			0.55	0.137932252711634\\
			0.56	0.140336199389228\\
			0.57	0.142699923295872\\
			0.58	0.14507930871737\\
			0.59	0.147446935953226\\
			0.6	0.149779424196541\\
			0.61	0.152135629777975\\
			0.62	0.154509956597198\\
			0.63	0.156833914647433\\
			0.64	0.159152994190537\\
			0.65	0.161470009208212\\
			0.66	0.163766107006464\\
			0.67	0.166096408533233\\
			0.68	0.168440691890017\\
			0.69	0.170678311964149\\
			0.7	0.172998924667977\\
			0.71	0.175249770781372\\
			0.72	0.177525225808243\\
			0.73	0.179792675780683\\
			0.74	0.182095313999853\\
			0.75	0.184340609291412\\
			0.76	0.186581001044066\\
			0.77	0.188805983105536\\
			0.78	0.191034598612346\\
			0.79	0.193279719720002\\
			0.8	0.195504859123132\\
			0.81	0.197695243536352\\
			0.82	0.199899438190622\\
			0.83	0.202089283186602\\
			0.84	0.204283569282951\\
			0.85	0.206458181541426\\
			0.86	0.20860535972673\\
			0.87	0.210761053107208\\
			0.88	0.212920105877607\\
			0.89	0.215059259683089\\
			0.9	0.217179510009315\\
			0.91	0.219318605098987\\
			0.92	0.221425122669127\\
			0.93	0.2235530256428\\
			0.94	0.225631738498246\\
			0.95	0.227715992757901\\
			0.96	0.229811204509025\\
			0.97	0.231895682004869\\
			0.98	0.23392296733761\\
			0.99	0.235967637429197\\
			1	0.238024924849803\\
			1.01	0.240081569744368\\
			1.02	0.242058846911967\\
			1.03	0.244110447740564\\
			1.04	0.246084106787463\\
			1.05	0.248094702136778\\
			1.06	0.250055054988388\\
			1.07	0.252042404948297\\
			1.08	0.25401510877129\\
			1.09	0.255982574843397\\
			1.1	0.257908766705282\\
			1.11	0.259806712517919\\
			1.12	0.261723850450288\\
			1.13	0.263632307291065\\
			1.14	0.265564088501987\\
			1.15	0.267407640788178\\
			1.16	0.269305908769802\\
			1.17	0.271187012300665\\
			1.18	0.273015010653568\\
			1.19	0.27480138707454\\
			1.2	0.276647589273077\\
			1.21	0.278470969498784\\
			1.22	0.280281026786595\\
			1.23	0.282033254775489\\
			1.24	0.283810745753966\\
			1.25	0.285619459575786\\
			1.26	0.287326131468151\\
			1.27	0.289070276029498\\
			1.28	0.29075571344277\\
			1.29	0.292509824578434\\
			1.3	0.294152127082875\\
			1.31	0.295834011994937\\
			1.32	0.297493938909542\\
			1.33	0.2991819093623\\
			1.34	0.300751945869386\\
			1.35	0.302381773331112\\
			1.36	0.303939325248925\\
			1.37	0.305529917947254\\
			1.38	0.307130476158997\\
			1.39	0.308618305561746\\
			1.4	0.310186963599191\\
			1.41	0.311668436396021\\
			1.42	0.313156162151994\\
			1.43	0.314599629277902\\
			1.44	0.316059147947167\\
			1.45	0.31742387104179\\
			1.46	0.318877421593978\\
			1.47	0.320204661781847\\
			1.48	0.321534361732743\\
			1.49	0.322859822432114\\
			1.5	0.324172051943135\\
			1.51	0.325400855907222\\
			1.52	0.326630967932068\\
			1.53	0.327822093982709\\
			1.54	0.32900028519477\\
			1.55	0.330066115324913\\
			1.56	0.331123864956177\\
			1.57	0.332055094110013\\
			1.58	0.332964712079097\\
		};
		\addlegendentry{Theory}
		
	\end{axis}
\end{tikzpicture}%
		\caption{}
	\end{subfigure}
	\caption{(a) Optimal $ \bv{v} $ for $ \alpha = 0.5 $ vs $ C $. (b) $ R_{T_{\alpha}} (C) $ vs $ C $ for $ \alpha=0.5 $.}
	\label{figure:tito_hamming_vs_C_alpha_0p5}
\end{figure}

\subsection{Examples}
Now let us consider two special cases.
\subsubsection{\acrshort{bms}}
Assume that the channel from $ \rv{X} $ to $ \rv{Y} $ is a \acrshort{bms} channel. Let $ \bv{z} $ be an $ m $-ary probability vector and $ G_m $ be the $ m\times  m $ anti-diagonal matrix with unit entries. The respective transition matrix in this case is $ \mat{T} = [\bv{z}, G_m \bv{z}] $. Note that
\begin{equation}
	G_m T = [\mat{G}_m \bv{z}, \mat{G}_m \mat{G}_m \bv{z}]  = [\bv{z} \mat{G}_2 \bv{z}] = T G_2.
\end{equation}
Therefore, $ \mat{T} $ is input symmetric stochastic matrix with input symmetry group $ \group{G}_i $ of order $ 2 $. Thus, since the only binary-input binary-output symmetric channel is  a \acrshort{bsc}, combining with \autoref{theorem:input_symmetric_transition_matrices}, we recover the following result from \cite{Sutskover2005}.
\begin{corollary}[{\cite[Theorem 2]{Sutskover2005}}]
	Given that the channel from $ \rv{X} $ to $ \rv{Y} $ is a \acrshort{bms}, then \acrshort{bsc} channel from $ \rv{X} $ to $ \rv{W} $ maximizes $ I(\rv{W};\rv{Y}) $.
\end{corollary}
The latter result can also be deduced from \cite{Chayat1989}.

\subsubsection{\acrfull{tito} Circulant Matrix}

The general \acrshort{tito} Circulant Matrix is defined as follows:
\begin{equation} \label{eq:tito_transition_matrix}
	T = 
	\begin{pmatrix}
		1-\alpha-\beta & \alpha & \beta \\
		\beta & 1-\alpha-\beta  & \alpha \\
		\alpha & \beta & 1-\alpha - \beta
	\end{pmatrix}.
\end{equation}

We can further ask if there are values of  $ C $ such that $ R(C) $ can be achieved with $ \rv{W}  $ taking at most two points. The following corollary states the opposite. 
\begin{corollary} \label{corollary:ssib_symtito_cardinality}
	The minimum cardinality of $ \rv{W} $ that achieves $ R(C) $ is exactly 3 for $ C \neq 0 $.
\end{corollary}
The proof of this corollary relegated to Appendix \ref{section:proof_sy_tito_cardinality}.

We proceed to verify \autoref{theorem:ssib_hamming} via numerical optimization for $ n =3 $. Since $ \rv{V} $ is independent of the choice of $ \alpha $, we freeze $ \alpha = 0.5 $ and compare it with respect to the value of $ C $. \autoref{figure:tito_hamming_vs_C_alpha_0p5} shows the probability vector $ \rv{V} $ and $ R_{T_\alpha}(C) $ for various values of $ \alpha $. We observe that the numerical optimization agrees with theoretical arguments of   \autoref{theorem:ssib_hamming}.

\section{The Symmetric Privacy Funnel}

In this section we consider a special symmetric setting for the \acrshort{pf} problem \eqref{eq:sspf_problem_definition2}  for which the transition matrix from $ \rv{X} $ to $ \rv{Y} $ is an input symmetric stochastic matrix as  defined in \autoref{definition:symmetric_group}.
\begin{theorem} \label{theorem:sspf_modulo_additive}
	Let $ T $ be an input symmetric stochastic matrix with input symmetry group $ \group{G}_i $ of order $ n $, and $ \rv{X} $ be a uniformly distributed random variable. Let $ (G_1=I,G_2,\dots, G_n ) \in \group{G}_i $. Furthermore, denote by $ (\bv{p}^*,\lambda^*) $ a pair for which
	\begin{equation} \label{eq:sspf_argmax_phi}
		\phi(\bv{u},\lambda^*) = \phi(\bv{p}^*,\lambda^*) \geq \phi(\bv{p},\lambda^*) \quad \forall \bv{p} \in \simplex{n}.
	\end{equation}
	Then, for every $ C \leq C^* \triangleq \log n - h_n(\bv{p}^*) $, the transition matrix from $ \rv{W} $ to $ \rv{X} $, given by 
	\begin{equation} \label{eq:sspf_Bmat}
		\mat{B} = 
		\begin{pmatrix}
			\bv{p}^* & \mat{G}_2 \bv{p}^* & \cdots & \mat{G}_n \bv{p}^* & \bv{u}
		\end{pmatrix},
	\end{equation}
	achieves \eqref{eq:sspf_problem_definition}. Moreover,
	\begin{equation}
		\reflectbox{R}_{\scaleto{\mathsf{P}_{\rv{X}\rv{Y}}}{4pt}}(C) = C  \cdot \frac{\log n - h_n(\mat{T} \bv{p}^*)}{\log n - h_n(\bv{p}^*)}.
	\end{equation}
	Also, \eqref{eq:sspf_Bmat} implies that the transition matrix from $ \rv{X} $  to $ \rv{W} $ is a class of noisy $ n $-ary symmetric erasure channel.
\end{theorem}
Note that the optimization procedure in \eqref{eq:sspf_argmax_phi} is performed once for every $ C \in [0,\log_2 n - h_n(\bv{p}^*) ] $. Moreover, for $ C  \in [0,\log_2 n - h_n(\bv{p}^*) ]  $, the optimal test-channel from $ \rv{X} $ to $ \rv{W} $ is no longer symmetric as we show using an example.

The proof of this theorem is postponed to Appendix \ref{section:proof_sspf_modulo_additive}.

We now provide some examples that illustrate \autoref{theorem:sspf_modulo_additive}.

\subsection{Examples}
We begin with the simplest scenario where  $ \rv{X} $ is a binary random variable. Plugging this choice in \autoref{theorem:sspf_modulo_additive} and noting that $ \bv{p}^*=\bv{e} $ in this case, results in the following corollary.
\begin{corollary}
	Assume that the channel from $ \rv{X} $ to $ \rv{Y} $ is a \acrshort{bms}, then, \acrshort{bec} test-channel $ \mathsf{P}_{\rv{W}|\rv{X}} $ with parameter $ \epsilon = 1-C $ minimizes $ I(\rv{Y};\rv{W}) $ subject to $ I(\rv{X};\rv{W}) = C $.
\end{corollary}
Note that this result recovers \cite[Theorem 1]{Sutskover2005}, but here with only one-sided symmetry restriction.

We further illustrate \autoref{theorem:sspf_modulo_additive} using numerical optimization for a particular choice of the channel from $ \rv{X} $ to $ \rv{Y} $ being a symmetric \acrshort{tito} with parameters $ (\alpha,\beta) = (0.1,0.05) $, as defined in \eqref{eq:tito_transition_matrix}. For this  choice of channel parameters, $ C^* =  0.59 $. In \autoref{figure:sspf_symtito_alpha_0p1_beta_0p05_RvsC} we compare the results of global optimization solution of  \eqref{eq:sspf_problem_definition} versus the method described in \autoref{theorem:sspf_modulo_additive} for various values of $ C $. We observe that our results from \autoref{theorem:sspf_modulo_additive} agree with the brute-force numerical optimization for all values of $ C \in [0,C^*]  $. For values greater than $ C^* $ the theoretical curve is restricted to input symmetric transition matrices  from $ \rv{X} $ to $ \rv{W} $. In this region of link capacity,  the numerical optimization achieves lower rates. By carefully observing the numerical solution, one can notice that the optimal  test-channel in this region is no longer input symmetric.

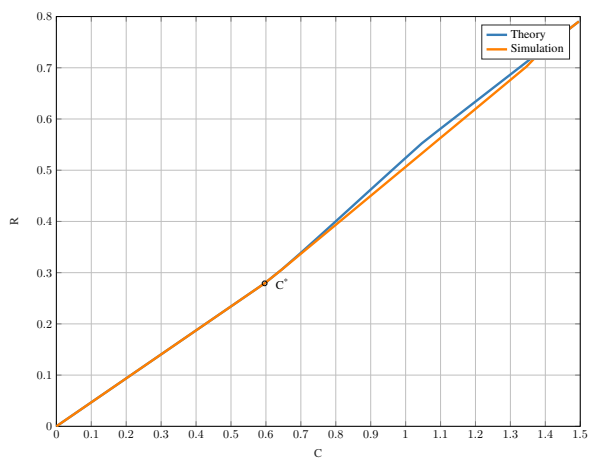
\begin{figure}[b!]
	\centering
%
%
\definecolor{mycolor1}{rgb}{0.21569,0.49412,0.72157}%
\definecolor{mycolor2}{rgb}{1.00000,0.49804,0.00000}%
\begin{tikzpicture}[scale=0.45]
	
	\begin{axis}[%
		width=6.028in,
		height=4.754in,
		at={(1.011in,0.642in)},
		scale only axis,
		xmin=0,
		xmax=1.5,
		xlabel style={font=\color{white!15!black}},
		xlabel={C},
		ymin=0,
		ymax=0.8,
		ylabel style={font=\color{white!15!black}},
		ylabel={R},
		axis background/.style={fill=white},
		xmajorgrids,
		ymajorgrids,
		legend style={legend cell align=left, align=left, draw=white!15!black}
		]
		\addplot [color=mycolor1, line width=2.0pt]
		table[row sep=crcr]{%
			0	0\\
			0.0499999999999999	0.0234133639007899\\
			0.0999999999999999	0.0468263762784363\\
			0.15	0.070239389575562\\
			0.2	0.0936524032088098\\
			0.25	0.117065347492647\\
			0.300000000001062	0.140478430377192\\
			0.350000000000118	0.163891444568642\\
			0.4	0.187304458784675\\
			0.45	0.210717472760513\\
			0.5	0.234130486767231\\
			0.549999999999996	0.257543441865843\\
			0.596202803346836	0.279178598200568\\
			0.64620280334684	0.306627220394201\\
			0.69620280334004	0.336442232439104\\
			0.746202803346787	0.366973420397012\\
			0.796202803345997	0.397832391801042\\
			0.846202803347513	0.428836389505368\\
			0.896202803347353	0.459873781189727\\
			0.946202803346833	0.49085409838481\\
			0.996202803343182	0.521713701330983\\
			1.04620280334681	0.552379369485204\\
			1.09620280334684	0.579153080473284\\
			1.14620280334684	0.605569409582508\\
			1.19620280334684	0.631985738691733\\
			1.24620280334684	0.658402067800957\\
			1.29620280334684	0.684818396910181\\
			1.34620280334684	0.711234726019406\\
			1.39620280334684	0.73765105512863\\
			1.44620280334684	0.764067384237854\\
			1.49620280334684	0.790483713347078\\
		};
		\addlegendentry{Theory}
		
		\addplot [color=mycolor2, line width=2.0pt]
		table[row sep=crcr]{%
			0	0\\
			0.0499874358837442	0.0234081970564289\\
			0.0999749892041695	0.0468153760207062\\
			0.149966785037927	0.0702245452447915\\
			0.199966825737283	0.0936375764282787\\
			0.249961112300046	0.117047914039205\\
			0.299933058941223	0.14044779112612\\
			0.349913171691581	0.163851492202362\\
			0.399927935912668	0.187271419617207\\
			0.44992646397477	0.210683744304666\\
			0.499933500887677	0.234100053347381\\
			0.54988736405373	0.257491462998469\\
			0.596113240387037	0.279137569910844\\
			0.64611189486527	0.306384929727258\\
			0.695563275003478	0.334353607256732\\
			0.746111754398221	0.362942791902162\\
			0.796084733583848	0.391206482725555\\
			0.846028025810318	0.419453385958484\\
			0.896082329134592	0.447774592866752\\
			0.946138325554235	0.476074068431975\\
			0.996047894743738	0.504301894742702\\
			1.04609253733721	0.53260611607382\\
			1.0961056771215	0.560970943605068\\
			1.14623265237884	0.589243661962113\\
			1.19625732620697	0.617536582799755\\
			1.24610710170711	0.645730590601524\\
			1.29623369659963	0.674081162400264\\
			1.34613677570611	0.702305323530622\\
			1.39597123805401	0.737530801260316\\
			1.44608913292301	0.764009416873294\\
			1.49609086480762	0.79042666051194\\
		};
		\addlegendentry{Simulation}
		
		\addplot [color=black, draw=none, mark=o, mark options={solid, black}, forget plot]
		table[row sep=crcr]{%
			0.596202803346836	0.279137569910844\\
		};
		\node[right, align=left]
		at (axis cs:0.616,0.279) {$\text{C}^\text{*}$};
	\end{axis}
	
	\begin{axis}[%
		width=7.778in,
		height=5.833in,
		at={(0in,0in)},
		scale only axis,
		xmin=0,
		xmax=1,
		ymin=0,
		ymax=1,
		axis line style={draw=none},
		ticks=none,
		axis x line*=bottom,
		axis y line*=left,
		legend style={legend cell align=left, align=left, draw=white!15!black}
		]
	\end{axis}
\end{tikzpicture}%
	\caption{Optimal $ R $ for $ \alpha = 0.1 $, $ \beta = 0.05 $ vs $ C $}
	\label{figure:sspf_symtito_alpha_0p1_beta_0p05_RvsC}
\end{figure}


\section{Outlook}
As said, the \acrlong{ib} and Privacy Funnel are  two dual optimization problems which have been applied in a variety of emerging applications such as Deep Neural Networks, Privacy Algorithms, and design of Polar Codes  \cite{Goldfeld2020}. It also  interesting to consider rather more classical use-cases, i.e, multi-user channel capacity and Noisy Source Coding problems.
A comprehensive summary of the different relations between  the \acrshort{ib} and Privacy Funnel problems has been presented in \cite{Asoodeh2020}.

\section*{Acknowledgment}
This work has been supported by the European Union's 
Horizon 2020 Research
And Innovation Programme, grant agreement no. 694630.

\bibliographystyle{IEEEtran}
\bibliography{bibliography}

\newpage
\appendix
\subsection{Proof of \autoref{proposition:equivalence_of_ssib_and_ww1975}}\label{section:proof_equivalence_of_ssib_and_ww1975}
Since $ P_{\rv{X}\rv{Y}} $ is fixed, then $ H(\rv{X}) $ and $ H(\rv{Y}) $ are determined. It follows that the constraint $ I(\rv{X};\rv{W})  \leq C $ is equivalent to $ H(\rv{X}|\rv{W}) \geq h_n(\bv{q}) - C  $. 
In addition, $ I(\rv{Y};\rv{W})  = H(\rv{Y}) - H(\rv{Y}|\rv{W})$. Hence \eqref{eq:ssib_problem_definition} can be rewritten as follows
\begin{align} \label{eq:ssib_problem_definition3}
	R_T(\bv{q},C) &= H(\rv{Y})- \min_{P_{\rv{X}|\rv{W}}:H(\rv{X}|\rv{W}) \geq h_{n}(\bv{q}) -  C}   H(\rv{Y}|\rv{W}) \\
	&=h_m(T \cdot \bv{q}) - F_T(\bv{q},h_n(\bv{q}) - C).
\end{align}
This concludes the proof of \autoref{proposition:equivalence_of_ssib_and_ww1975}.

\subsection{Proof of the main result}\label{section:proof_of_circulant_case}
The equivalence of the \acrshort{ib} and \acrshort{ceb} problems has been shown in \autoref{proposition:equivalence_of_ssib_and_ww1975}. Therefore, we will consider the equivalent \acrshort{ceb} formulation in our proof. We begin with stating the main utility result of \cite{Witsenhausen1975}.
\begin{lemma}[{\cite[Theorem 4.1]{Witsenhausen1975}}] \label{lemma:ww1975_thm4p1}
	Let $ \phi(\bv{p},\lambda) \triangleq h_m(\mat{T}\bv{p}) - \lambda h_n(\bv{p}) $ and let $ \psi(\cdot,\lambda) $ be the lower convex envelope on $ \simplex{n} $ of $ \phi(\cdot,\lambda) $. Then
	\begin{enumerate}
		\item \begin{align}
			F_T(\bv{q},x) &= \max \{\psi(\bv{q},\lambda) + \lambda x | 0\leq \lambda \leq 1 \}, \\
			F_T(\bv{q},0) &= H(\rv{Y}|\rv{X}), \\
			F_T(\bv{q},H(\rv{X})) &= H(\rv{Y}) = h_m(\mat{T}\bv{q}).
		\end{align}
		\item  If a point of the graph of $ \psi(\cdot,\lambda) $ is the convex combination of $ k $ points of the graph of $ \phi(\cdot,\lambda) $ with arguments $ \bv{p}_{\alpha} $ and weights $ w_{\alpha} $ , ($\alpha =1,\dots,k$), then
		\begin{equation}
			F_T\left( \sum_{\alpha=1}^k w_\alpha \bv{p}_\alpha, \sum_{\alpha=1}^{k}w_\alpha h_n(\bv{p}_{\alpha}) \right)  =\sum_{\alpha=1}^k w_a h_m(T\bv{p}_{\alpha}).
		\end{equation}
		\item  If for some $ \bv{w} $ and $ \lambda $, $ \phi(\bv{q},\lambda) = \psi(\bv{q},\lambda) $, this corresponds to a line supporting the graph og $ F_T(\bv{q},\cdot) $ at the endpoint $ x = h_n(\bv{q}) $.
	\end{enumerate}
\end{lemma}
By \autoref{lemma:ww1975_thm4p1}, evaluation of $ F_T(\bv{q},x) $ goes through the analysis of $ \phi(\bv{p},\lambda) $. We proceed with deriving some properties of $ \phi(\bv{p},\lambda) $ for input symmetric matrices $ T $.
\begin{proposition} \label{proposition:phi_p_lambda_input_symmetric}
	Let $ T $ be input symmetric stochastic matrix with input symmetry group $ \group{G}_i $. Then, for every $ \mat{G} \in \group{G}_i $
	\begin{equation}
		\phi(\mat{G}\bv{p},\lambda) = \phi(\bv{p},\lambda) .
	\end{equation}
\end{proposition}
\begin{proof}
	Utilizing the symmetry property of the entropy function
	we have
	\begin{align}
		\phi(\mat{G}\bv{p},\lambda)
		&=  h_m(\mat{T} \mat{G} \bv{p}) - \lambda h_n( \mat{G} \bv{p}) \\
		&=  h_m(\mat{\Pi} \mat{T} \bv{p}) - \lambda h_n( \bv{p}) \\
		&=  h_m( \mat{T} \bv{p}) - \lambda h_n( \bv{p}) \\
		&= \phi(\bv{p},\lambda).
	\end{align}
\end{proof}

Now, let $ \bv{p}^* $ be the minimizer of $ \phi(\cdot,\lambda) $ over $ \bv{p} \in \simplex{n} $, i.e.,
\begin{equation}
	\bv{p}^* \triangleq \argmin_{\bv{p} \in \simplex{n}} \phi(\bv{p},\lambda).
\end{equation}
By \autoref{proposition:phi_p_lambda_input_symmetric} and the assumption that $ \group{G}_i $ is of order $ n $, we have
\begin{equation}
	\phi(\mat{G}_\alpha\bv{p}^*,\lambda) = \phi(\bv{p}^*,\lambda)  \qquad \forall \mat{G}_{\alpha} \in \group{G}_i, \alpha \in \{1,\dots,n\}.
\end{equation}
Further, denote $ \bv{p}_{\alpha} \triangleq \mat{G}_{\alpha} \bv{p}^* $ and consider a specific weights' vector $ \bv{w} = \{w_{\alpha}\}_{\alpha=1}^n $. Now, let $ \psi(\cdot ,\lambda) $ be the lower convex envelope  of $ \phi(\cdot,\lambda) $.   We obtain
\begin{equation*}
	\psi\left(\sum_{\alpha=1}^n w_a \bv{p}_{\alpha},\lambda \right) = \sum_{\alpha=1}^n w_a \phi(\bv{p}_{\alpha},\lambda) = \phi(\bv{p}^*,\lambda),
\end{equation*}
since $ \phi(\bv{p}^*,\lambda) $ is the minimum of $ \phi(\cdot,\lambda) $ over $ \simplex{n} $.
Therefore, by \autoref{lemma:ww1975_thm4p1}, it follows that
\begin{align}
	&F_T \left(\sum_{\alpha=1}^n w_a \bv{p}_{\alpha},\sum_{\alpha=1}^n w_a  h_n\left( \bv{p}_{\alpha}\right) \right)  \\
	&=\sum_{\alpha=1}^k w_a h_m(T\bv{p}_{\alpha}) \\
	&= h_m(T\bv{p}^{\star}) \\
	&= F_T \left(\sum_{\alpha=1}^n w_a G_{\alpha} \bv{p}^{\star},h_n\left( \bv{p}^{\star}\right) \right).
\end{align}

To this end, we have chosen $ \lambda \in (0,1) $, then obtained $ \bv{p}^* $ and found a solution for $ F_T(\bv{q},x) $ where $ \bv{q} = \sum_{\alpha=1}^n w_\alpha G_{\alpha} \bv{p}^* $ and $ x = h_n(\bv{p}^{\star}) $ for any $ \bv{w} \in \simplex{n} $.  
Let $ \set{S} $ be the set defined in \eqref{eq:definition_of_the_set_q_C}.
Thus, for every $ (\bv{q},x) \in \set{S} $, $ F_T(\bv{q},x) $ is achieved with input symmetric transition matrix from $ \rv{W} $ to $ \rv{X} $. An $ n \times  n $ input symmetric matrix is a circulant matrix.  This completes the proof of the main result.

\subsection{Proof of \autoref{corollary:ssib_symtito_cardinality}}\label{section:proof_sy_tito_cardinality}
In general, for every $ \bv{p}_1 = \bv{p}_\lambda = (p_\lambda,r_\lambda,1-p_\lambda-r_\lambda)^T $ that minimizes $ \phi(\bv{p},\lambda) $, there exists
\begin{equation*}
	\bv{p}_2 \mkern-4mu = \mkern-4mu \Pi_2 \bv{p}_1  \mkern-4mu = \mkern-4mu
	\begin{pmatrix}
		r_\lambda \\ 1-p_\lambda-r_\lambda \\ p_\lambda
	\end{pmatrix},
	\
	\bv{p}_3 \mkern-4mu =  \mkern-4mu \Pi_3 \bv{p}_1 \mkern-4mu = \mkern-4mu
	\begin{pmatrix}
		\mkern-2mu	1-p_\lambda-r_\lambda \\ p_\lambda \\ r_\lambda \mkern-2mu
	\end{pmatrix}.
\end{equation*}
Assume that $ |\set{Z}|=2 $, therefore either $ \bv{p}_1 = \bv{p}_2 $ or $ \bv{p}_1 = \bv{p}_3 $ or $ \bv{p}_2 = \bv{p}_3 $. Any of this conditions imply that $ \bv{p}_1 = \bv{p}_2 = \bv{p}_3 = \bv{u}_3 $, and $ |\set{W}| = 1 $. This further implies that $ |\set{\rv{W}}|\neq 2 $ and if $ |\set{W}| = 1 $, than $ C = \log 3 - h(\bv{u}_3) = 0 $.

\subsection{Proof of \autoref{theorem:sspf_modulo_additive}}\label{section:proof_sspf_modulo_additive}
Define the set $ \set{S} $ as the collection of points $ (\bv{p},h_n(\bv{p}),h_n(\mat{T}\bv{p}) $ for every $ \bv{p} \in \simplex{n} $. Let $ \set{C} $ denote the convex hull of $ \set{S} $. In similar manner to \cite{Witsenhausen1975}, one can show that $ \set{C} $ is determined by the following set of triples $ (\bv{p},\xi,\eta) $.
\begin{align*}
	\bv{p} &= \sum_{\alpha=1}^{n+1} w_{\alpha} \bv{p}_{\alpha}, \\
	\xi &= \sum_{\alpha=1}^{n+1} w_{\alpha} h_n( \bv{p}_{\alpha}),\\
	\eta &= \sum_{\alpha=1}^{n+1} w_{\alpha} h_n( \bv{p}_{\alpha}),
\end{align*}
for all $ \bv{w} \in \simplex{n+1} $ and $ \bv{p}_{\alpha} \in \simplex{n} $.
Furthermore, $ \reflectbox{F}_{\scaleto{\mat{T}}{4pt}}(\bv{q},x) $ is the maximum of all $ \eta $ for which $ \bv{q} = \bv{p} $, $ x = \xi $ belong to $ \set{C} $.

In a very similar manner to \cite{Witsenhausen1975} and the proof of \autoref{theorem:input_symmetric_transition_matrices}, our goal is to find the upper convex envelope of $ \phi(\bv{p},\lambda) $ using at most $ n+1 $ points. If 
\begin{equation} \label{eq:sspf_proof_phi_inequality_condition}
	\phi(\bv{u},\lambda) \geq \phi(\bv{p},\lambda) \qquad \forall \bv{p} \in \simplex{n}
\end{equation}
with equality only for $ \bv{p} = \bv{u} $, then we are done, since in this case the only relavant point is $ \bv{p} = \bv{u} $ and $ R=C=0 $ in this case. Assume that there exists $ \lambda^{\star} $ for which the equality in \eqref{eq:sspf_proof_phi_inequality_condition} also holds for $ \bv{p}^* $.
Since $ \mat{T} $ is input symmetric with input symmetry group of order $ n $, and the symmetry property of $ \phi(\bv{p},\lambda) $ as in \autoref{proposition:phi_p_lambda_input_symmetric}, then there are $ n $ points such that $ \bv{p}_k = \mat{G}_k \bv{p}^*  $ and
\begin{equation}
	\phi(\bv{u},\lambda^*) = \phi(\bv{p}_k,\lambda^*)  \geq \phi(\bv{p},\lambda^*) \quad \forall \bv{p} \in \simplex{n},  k \in \{1,\dots,n\}.
\end{equation}
Thus, the upper concave envelope of $ \phi(\bv{p},\lambda^*) $  consists of the $ n+1 $ points
\begin{equation}
	\bv{p}_1,\bv{p}_2,\dots,\bv{p}_n,\bv{p}_{n+1} = \bv{u}.
\end{equation}
Note, that using this points we can construct $ (\bv{p},\xi,\eta) = (\bv{u},x,\reflectbox{F}(x)) $ as  follows
\begin{align}
	\bv{u} &= \sum_{k=1}^{n+1} w_k \bv{p}_k , \\
	x &= \sum_{k=1}^{n+1} w_k h_n(\bv{p}_k) 
	= (1-\epsilon) h_n(\bv{p}^*) + \epsilon \log_2 n, \\
	\reflectbox{F}(x) &= \sum_{k=1}^{n+1} w_k h_n(\mat{T} \bv{p}_k) 
	= (1-\epsilon) h_n(T  \bv{p}^*) + \epsilon \log_2 n.
\end{align}
Therefore, 
\begin{equation}
	\epsilon = \frac{x-h_n(\bv{p}^*)}{\log_2 n - h_n(\bv{p}^*)}.
\end{equation}
Since $ \epsilon \geq 0 $, this will be valid for $ x \geq h_n(\bv{p}^*) $, or in our terminology, for $ C < \log_2 n - h_n(\bv{p}^*) $. This completes the proof.


\end{document}